\documentclass[11pt]{article}

\usepackage[margin=1in]{geometry}
\usepackage{amsmath,amssymb}
\usepackage{bm}
\usepackage[utf8]{inputenc}
\usepackage[T1]{fontenc}
\usepackage{lmodern}
\usepackage[hidelinks]{hyperref}

\newtheorem{proposition}{Proposition}

\newenvironment{proof}{\medskip\noindent\textit{Proof.}}{\hfill$\square$\medskip}

\numberwithin{equation}{section}

\begin{document}

\title{The Imaginary Component of the Weak Momentum:\\ Spectral Vanishing and Gauge Invariance}

\author{A. N. Khondker\thanks{Email: \texttt{akhondke@clarkson.edu}. ORCID: 0009-0003-1843-4480.}\\[4pt]
\small Department of Electrical and Computer Engineering,\\
\small Clarkson University, Potsdam, New York 13699, USA}

\date{}

\maketitle

\begin{abstract}
We establish a rigorous functional-analytic framework for the
imaginary component $p_I$ of the complex weak momentum
$p_Q = p_R + ip_I = \frac{\hbar}{i}\nabla\ln\psi$. For every
wavefunction in the Sobolev space $H^1(\Omega)$, including those
with nodal sets, we prove that the expectation value of $p_I$
vanishes identically, as a consequence of the regularity of
$\rho = |\psi|^2$ and without appeal to the polar form
$\psi = Re^{iS/\hbar}$. We further show that
\textnormal{(i)} the Bohm quantum potential arises entirely from
$p_I$ and satisfies
$\langle V_{\mathrm{qu}}\rangle = \langle p_I^2\rangle/2m$, from
which the known Fisher information relation follows as a corollary;
\textnormal{(ii)} the group velocity of a Bloch band is the
density-weighted average of $p_R/m$, with no contribution from $p_I$;
\textnormal{(iii)} the decomposition is gauge-invariant under
magnetic minimal coupling, the vector potential acting only on
$p_R$; and
\textnormal{(iv)} the results extend to curved Riemannian manifolds.
\end{abstract}

\noindent\textbf{Keywords:} weak momentum, spectral vanishing, gauge invariance, quantum potential, Fisher information, Sobolev regularity

\section{Introduction}
Nonlinear reformulations of quantum mechanics uncover structural features hidden within linear operator expressions~\cite{Khondker2025}. In prior work, a nonlinear momentum decomposition was introduced, where the linear operator $-i\hbar\nabla$ arises from the sum of real and imaginary functional transforms. While that work successfully demonstrated the existence of the decomposition, the functional-analytic implications of this structure were not established in that work. Specifically, the question of whether the imaginary component preserves the spectral integrity of the momentum operator requires a rigorous treatment of boundary functionals.

The geometric structure of quantum state space has been a 
focus of ongoing mathematical research, with foundational 
contributions by Kibble~\cite{Kibble}, Ashtekar and 
Schilling~\cite{Ashtekar}, and Brody and 
Hughston~\cite{Brody} establishing symplectic and Riemannian 
structures on the manifold of quantum states. The present 
work addresses a complementary question: not how the space 
of states is geometrized, but how the momentum operator 
itself factorizes according to the amplitude--phase structure 
of a single wavefunction.

Several related geometric programs inform but do not resolve the questions addressed here. Provost and Vall\'{e}e~\cite{provost1980} introduced a Riemannian metric on manifolds of quantum states; Berry~\cite{berry1984} identified the geometric phase arising from holonomy on parameter-space bundles, and the framework of information geometry~\cite{amari2000} equips probability manifolds with the Fisher metric, recovering the Fubini--Study geometry on projective Hilbert space~\cite{Brody,amari2000}. Fisher information has also been shown to play a foundational role in quantum mechanics directly: Frieden~\cite{Frieden1998} developed a unifying principle of physics based on extremizing Fisher information, and Reginatto~\cite{Reginatto1998} and Hall and Reginatto~\cite{HallReginatto2002} derived the Schr\"odinger equation --- including the quantum potential --- from a variational principle of minimum Fisher information. These works establish a connection between Fisher information and the quantum potential at the level of the equation of motion. The present work establishes a complementary and more elementary connection: the identity $\langle V_{\mathrm{qu}}\rangle = \langle p_I^2\rangle/2m$, proved directly from the operator structure of the weak momentum decomposition via Sobolev integration by parts, without any variational principle and for all $H^1$ states including those with nodal sets. The known Fisher information relation $\langle V_{\mathrm{qu}}\rangle = \hbar^2\mathcal{I}[\rho]/8m$ then follows as a corollary. Taken together, these works connect Fisher information to quantum dynamics, but none addresses how the momentum operator itself factorizes according to the internal amplitude--phase structure of a single wavefunction. The present work fills this gap.

Weak values~\cite{Aharonov1988,Leavens2005,Matzkin2015,Hiley2012,Flack2018}, on the other hand, extend quantum expectation values through pre- and post-selection~\cite{Dressel2014,Aharonov2014}. Hydrodynamic formulations, beginning with Madelung and later developed by Takabayasi, associate the phase gradient with the probability current~\cite{Madelung1927,Wyatt2005,Takabayasi}. Later work established connections between weak values and locally averaged momentum flow~\cite{Wiseman2007}. The weak momentum value is a complex number assigned to the momentum observable through the process of preselection ($|i\rangle$) and post-selection ($|f\rangle$), defined by the formula $p_{w}=\langle f|\hat{p}|i\rangle/\langle f|i\rangle$ ~\cite{Aharonov1988,Dressel2014,Aharonov2014}. 
Recent weak-measurement experiments have enabled reconstruction of local momentum--flow fields consistent with hydrodynamic formulations of quantum dynamics~\cite{Kocsis2011,Mahler2016,Zhu2021}, providing empirical access to local momentum fields related to those studied here. The imaginary part of the momentum weak value is associated with the logarithmic gradient of the wavefunction's amplitude~\cite{Leavens2005,Matzkin2015}, which is the same object as $p_I$ studied here.

The purpose of the present work is threefold. First, we establish the functional-analytic well-posedness of the decomposition, proving that the imaginary component $p_I$ reduces to a divergence-generated boundary functional under minimal regularity assumptions on $\rho = |\psi|^2$, with its expectation value vanishing identically for all physically admissible states. Second, we prove the exact identity $\langle V_{\mathrm{qu}}\rangle = \langle p_I^2\rangle/2m$ via Sobolev integration by parts --- a result that does not require the polar form and holds for all $H^1$ states --- from which the Fisher information relation $\langle V_{\mathrm{qu}}\rangle = \hbar^2\mathcal{I}[\rho]/8m$ follows as a corollary. Third, we extend the result to Riemannian manifolds and to magnetic minimal coupling, establishing that the decomposition is intrinsic and gauge-invariant, independent of any particular representation of the wavefunction.

The paper is organized as follows. Section~2 reviews the nonlinear transform. Section~3 establishes the boundary-functional identity in Sobolev spaces and extends it to curved Riemannian manifolds. Section~4 develops the quantum potential and its relation to $\langle p_I^2\rangle/2m$, including verification for the infinite square well. Section~5 discusses canonical realizations and structural implications. Section~6 presents the gauge-invariant magnetic extension.

\section{Nonlinear Momentum Transform}
The linear quantum momentum operator $-i\hbar\nabla$ is a self-adjoint operator on $L^2(\Omega)$ with a well-defined spectral structure. The constructions introduced in~\cite{Khondker2025} do not modify this operator. Instead, they provide a \emph{state-dependent nonlinear decomposition of its action} on a given wavefunction $\psi$.

Define the complex local momentum field, also known as the complex weak value of the momentum~\cite{Aharonov1988,Leavens2005,Matzkin2015,Wiseman2007}:
\begin{equation}
p_Q(\mathbf r,t) =
\frac{\hbar}{i}\nabla \ln \psi(\mathbf r,t), \label{eq:ref2.1} 
\end{equation}
and split it into real and imaginary parts~\cite{Khondker2025}:
\begin{equation}
p_Q =p_R + i p_I,
\qquad
p_R = \Re e(p_Q),
\qquad
p_I = \Im m(p_Q). \label{eq:ref2.2}
\end{equation}

We then define nonlinear maps
\begin{equation}
\mathcal P_R[\psi]=p_R \psi, \qquad
\mathcal P_I[\psi]=p_I \psi. \label{eq:ref2.3}
\end{equation}

These maps depend nonlinearly on $\psi$ through the local fields 
$p_R$ and $p_I$. Although neither transform is linear, their sum reproduces the linear momentum operator $\mathcal P_Q$,
\begin{equation}
\mathcal P_R[\psi]+i\mathcal P_I[\psi]= \mathcal P_Q[\psi] = -i\hbar\nabla\psi. \label{eq:ref2.4}
\end{equation}

Throughout this work, we carefully distinguish between two related objects. The \emph{nonlinear transforms} $\mathcal{P}_R$ and $\mathcal{P}_I$ are operator-like maps acting on $\psi$. The \emph{local fields} $p_R(\mathbf{r},t)$ and $p_I(\mathbf{r},t)$ are the real-valued vector fields that these transforms multiply $\psi$. Statements about operator structure are expressed in terms of $\mathcal{P}_R$, $\mathcal{P}_I$, while those about local dynamics, the quantum potential, and the continuity equation are expressed in terms of $p_R$, $p_I$.

Thus, the nonlinear decomposition uncovers a geometric splitting of amplitude and phase flows that is implicit but not visible in the linear-operator formulation, aligning with hydrodynamic decompositions discussed by Takabayasi~\cite{Takabayasi} and the causal interpretation by Bohm~\cite{Bohm1952}. The local field $p_I/m$ coincides with what Nelson termed the osmotic velocity in stochastic mechanics~\cite{Nelson1966}, though the present framework makes no appeal to stochastic methods.

As a verification, when $\psi$ admits the polar form, the definitions reduce to the familiar expressions. The polar form:
\begin{equation}
\psi(\mathbf r,t)=R(\mathbf r,t)e^{iS(\mathbf r,t)/\hbar}, \qquad R\ge0, \label{eq:ref2.5}
\end{equation}
gives explicitly:
\begin{equation}
p_R = \nabla S, \qquad
p_I = -\hbar\nabla\ln R, \label{eq:ref2.6}
\end{equation}
identifying $p_R$ with the de~Broglie-Bohm momentum and $p_I$ 
with the logarithmic gradient of the probability amplitude.
The paper does not assume this form; it is shown here only 
to connect the abstract definitions to familiar notation~\cite{Bohm1952}.
The polar form is only used in the Appendix to verify equivalence 
with Bohm's original quantum potential expression.

In the \hyperref[sec:AppendixA]{Appendix}, we provide a full derivation of a Hamilton-Jacobi equation and the continuity equation, using only the complex local momentum field $p_R + ip_I$ defined in Eqs.~\eqref{eq:ref2.2}. The resulting formulation is structurally consistent with hydrodynamic, de Broglie-Bohm, weak-value constructions, yet does not introduce new observables and remains entirely within standard quantum mechanics. It is important to mention that while Ref.~\cite{Khondker2025} introduced the existence of the decomposition, this work establishes the functional-analytic well-posedness of the decomposition, as detailed in Section~3.

In the previous work~\cite{Khondker2025}, it was shown that 
$p_R$ is the local momentum field. A critical challenge in 
nonlinear decompositions is the behavior of the imaginary 
local field $p_I$ at nodal sets where $\psi = 0$. A naive 
pointwise evaluation of $\langle p_I \rangle$ is ill-defined 
at such points because $p_I \sim -\hbar\nabla\ln|\psi| \to 
-\infty$ as $|\psi| \to 0$. The following section examines 
a question of mathematical consistency: whether the imaginary 
component affects the expectation value of the Hermitian 
momentum operator under standard boundary conditions.

\paragraph{Interpretive note.}
The nonlinear transforms introduced here do not define linear operators on Hilbert space. Rather, they provide a state-dependent decomposition of the action of the linear momentum operator when evaluated on a given wavefunction. This distinction is used throughout Sections~3 and~4 to interpret the geometric structure of the decomposition without modifying the underlying operator framework.

\section{Structural Consistency and Functional-Analytic Framework} 

The vanishing of the imaginary weak-momentum expectation is a rigorous consequence of the underlying functional geometry. We establish this by placing the decomposition in the Sobolev regularity framework, ensuring that the splitting remains well-defined even in the presence of nodal sets. 

\subsection{Distributional Formulation and Nodal Regularity}

The key observation is that the product $\rho\, p_I$ remains 
locally integrable even at nodal sets. Using the density 
identity, Eq.~\eqref{eq:refA.density-id},  one obtains:
\begin{equation}
    \rho\, p_I \;=\; -\frac{\hbar}{2}\nabla\rho.
    \label{eq:ref3.1}
\end{equation}
The right-hand side is locally $L^1$ whenever 
$\rho \in W^{1,1}(\Omega)$, which holds for all 
$\psi \in H^1(\Omega)$ since $\rho = |\psi|^2$ inherits 
the Sobolev regularity of $\psi$ via the Cauchy--Schwarz 
inequality:
\begin{equation*}
    \int_\Omega |\nabla\rho|\,\mathrm{d}r
    \;\leq\; 2\,\big\||\psi|\big\|_{L^2(\Omega)}
             \big\|\nabla|\psi|\big\|_{L^2(\Omega)}
    \;<\; \infty.
\end{equation*}
This step would \emph{fail} if $\psi \notin H^1(\Omega)$, 
confirming that the $H^1$ assumption is sharp. The 
expectation integral $I_I$ therefore becomes:
\begin{equation}
    I_I \;=\; \int_\Omega \rho\, p_I\,\mathrm{d}r
        \;=\; -\frac{\hbar}{2}\int_\Omega 
              \nabla\rho\,\mathrm{d}r,
    \label{eq:ref3.2}
\end{equation}
where $\nabla\rho$ is interpreted as a weak gradient, 
valid for all $\rho \in W^{1,1}(\Omega)$, without 
requiring classical differentiability at nodal sets. 
Applying the Gauss--Green 
theorem~\cite{ChenTorresZiemer2009} to 
Eq.~\eqref{eq:ref3.2}:
\begin{equation}
    I_I \;=\; -\frac{\hbar}{2}
              \oint_{\partial\Omega}\rho\,\hat{n}\,\mathrm{d}S.
    \label{eq:ref3.3}
\end{equation}

This is the central structural identity of the paper. Three points deserve emphasis. First, Eq.~\eqref{eq:ref3.3} shows that $\langle p_I \rangle$ is determined entirely by the \emph{boundary data} of $\rho = |\psi|^2$, not by its interior structure; the imaginary momentum carries no bulk dynamical information. Second, for all physically admissible states --- those with decay at infinity, hard-wall Dirichlet conditions, or lattice periodicity --- the right-hand side vanishes identically (as verified in Section~5), establishing $\langle p_I \rangle = 0$ as a universal structural consequence of boundary topology. Third, the result is \emph{representation-independent}: the Gauss--Green theorem holds on any $W^{1,1}$ manifold with Lipschitz boundary~\cite{ChenTorresZiemer2009}, so the identity extends immediately to the Riemannian setting of Section~3.2.

The validity of Eq.~\eqref{eq:ref3.3} does not rely on classical differentiability at the nodes. The entire argument proceeds in the Sobolev sense: $\nabla\rho$ is interpreted as a weak gradient, and the Gauss--Green theorem is applied in its generalized form for $W^{1,1}$ vector fields~\cite{ChenTorresZiemer2009}.

\subsection{Riemannian Generalization}
\label{Theorem2}

The boundary-functional identity of Section~3.1 extends without 
modification to curved Riemannian manifolds, establishing that the 
vanishing of $\langle p_I\rangle$ is an intrinsic property of the 
Sobolev geometry of $\psi$, independent of any particular coordinate 
system or embedding in flat space.

Let $(M,g)$ be a smooth oriented Riemannian manifold with volume 
form $d\mu_g$ and Levi--Civita gradient $\nabla_g$. For 
$\psi \in H^1(M)$ with $\rho = |\psi|^2 \in W^{1,1}(M)$, 
the nonlinear decomposition extends naturally: the imaginary 
local field in the Riemannian setting is defined as 
$p_I^{(g)} = \mathrm{Re}(-\hbar\nabla_g\ln\psi)$, and the 
action of the kinetic momentum operator on $\psi$ decomposes as:
\begin{equation}
    -i\hbar\nabla_g\psi 
    \;=\; p_R^{(g)}\psi + i\,p_I^{(g)}\psi,
    \label{eq:Riem-decomp}
\end{equation}
where $p_R^{(g)} = \mathrm{Re}(\frac{\hbar}{i}\nabla_g\ln\psi)$ 
is the real local field in the Riemannian setting. The density 
identity of \hyperref[sec:AppendixA]{the Appendix} generalizes 
directly, since it follows from $\rho = |\psi|^2$ alone:
\begin{equation}
    \rho\, p_I^{(g)} = -\frac{\hbar}{2}\nabla_g\rho, 
    \qquad \rho = |\psi|^2,
    \label{eq:Riem-density}
\end{equation}
and integration over $M$ yields the Riemannian boundary functional:
\begin{equation}
    \int_M \rho\, p_I^{(g)}\,d\mu_g 
    = -\frac{\hbar}{2}\int_{\partial M} 
      \rho\,\hat{n}\,d\sigma_g,
    \label{eq:Riem-boundary}
\end{equation}
where $\hat{n}$ is the outward unit normal and $d\sigma_g$ is the 
induced boundary measure. The proof follows that of Section~3.1 
with the Euclidean gradient $\nabla$ replaced throughout by 
$\nabla_g$ and $\rho = |\psi|^2$ replacing $R^2$, with the 
Gauss--Green theorem applied in its Riemannian 
form~\cite{ChenTorresZiemer2009}. For all physically admissible 
states the boundary term on the right vanishes by the same 
topological argument as Section~3.1, establishing 
$\langle p_I^{(g)}\rangle = 0$ as a universal consequence of 
$W^{1,1}$ Sobolev regularity on any manifold with Lipschitz 
boundary.

\section{The Quantum Potential: Geometric Origin, Global Validity, and Fisher Information}

The nonlinear decomposition introduced in Section~2 carries a precise 
geometric meaning. The real local field $p_R$ governs 
dynamical transport --- it is the phase gradient responsible for the 
probability current and the Bohmian velocity. The imaginary local 
field $p_I = -\frac{\hbar}{2}\nabla\ln\rho$ governs amplitude deformation --- it 
is the logarithmic gradient of the probability density, responsible 
for the quantum potential and the osmotic flow. These two fields are 
structurally independent: the imaginary field $p_I$ carries no 
dynamical transport information, as its density-weighted average 
$\int_\Omega \rho\, p_I\, d\mathbf{r}$ vanishes for all states 
satisfying the boundary conditions of Section~3, while the real 
field $p_R$ carries the full physical momentum 
$\langle -i\hbar\nabla\rangle$. This geometric division of labor 
--- transport versus deformation, phase versus amplitude --- is the 
structural content of the decomposition, and the results of this section make it precise and rigorous.

\subsection{Quantum Potential}

Following the derivation in the Appendix, we apply the nonlinear transforms to evaluate the action of the 
Hamiltonian $\hat{H}$ on $\psi$, introducing mathematical operators 
$e_R(\mathbf{r},t)$ and $e_I(\mathbf{r},t)$ defined as 
the real and imaginary parts of $i\hbar(\partial_t\ln\psi)$. Their 
sum satisfies $e_R\psi + ie_I\psi = \hat{H}\psi$  (see ~\eqref{eq:refA.3}), and the real part 
yields:
\begin{equation}
    e_R \;=\; \frac{p_R^2}{2m} + V + V_{\mathrm{qu}},
    \label{eq:ref4.1}
\end{equation}
where $V_{\mathrm{qu}}$ is the structural quantum potential expressed solely
in terms of the imaginary local field $p_I$:
\begin{equation}
    V_{\mathrm{qu}} \;=\; \frac{\hbar\,\nabla\cdot p_I \;-\; p_I^2}{2m}.
    \label{eq:ref4.2}
\end{equation}
The structural form Eq.~\eqref{eq:ref4.2} is an exact algebraic 
identity with Bohm's original quantum potential 
$-\hbar^2\nabla^2 R/(2mR)$, as shown in 
\hyperref[sec:AppendixA]{the Appendix}.
\subsubsection{Global Validity and Distinction from the Pointwise 
de~Broglie--Bohm Equation}
\label{sec:global-validity}

Bohm's equation in its pointwise form~\cite{Bohm1952} reads:
\begin{equation}
    -\frac{\partial S}{\partial t} = \frac{(\nabla S)^2}{2m}
    + V(\mathbf{r}) - \frac{\hbar^2}{2m}\frac{\nabla^2 R}{R},
    \label{eq:Bohm-pointwise}
\end{equation}
which requires $R > 0$ at every point. At nodal sets 
$\{\psi = 0\}$ where $R = 0$, the term $\nabla^2 R/R$ diverges 
and Eq.~\eqref{eq:Bohm-pointwise} is not well-defined. The 
present formulation resolves this: $V_{\mathrm{qu}}$ enters only 
through its expectation value, which is finite and non-negative 
for all $\psi \in H^1(\Omega)$, as shown below.

Taking the expectation value 
of Eq.~\eqref{eq:ref4.2}:
\begin{equation}
    \langle V_{\mathrm{qu}}\rangle 
    = \frac{\hbar}{2m}\int_\Omega\rho\,\nabla\cdot p_I\,d\mathbf{r}
    - \frac{\langle p_I^2\rangle}{2m}.
    \label{eq:Vqu-split}
\end{equation}
Applying integration by parts to the first term, with boundary  
terms vanishing for all $\psi\in H^1(\Omega)$ by Section~3:
\begin{equation}
    \frac{\hbar}{2m}\int_\Omega\rho\,\nabla\cdot p_I\,d\mathbf{r}
    = -\frac{\hbar}{2m}\int_\Omega p_I\cdot\nabla\rho\,d\mathbf{r}. 
\end{equation}
Substituting the density identity 
$\nabla\rho = -\frac{2}{\hbar}\rho\,p_I$
from \hyperref[sec:AppendixA]{the Appendix}, which follows from 
Eqs.~\eqref{eq:ref2.1}--\eqref{eq:ref2.2} alone:
\begin{equation}
    -\frac{\hbar}{2m}\int_\Omega p_I\cdot\nabla\rho\,d\mathbf{r}
    = \frac{1}{m}\int_\Omega\rho\,p_I^2\,d\mathbf{r}
    = \frac{\langle p_I^2\rangle}{m}. \label{eq:ref4.6}
\end{equation}
As an immediate consequence of combining 
Eqs.~\eqref{eq:Vqu-split} and \eqref{eq:Vqu-pI}, the divergence 
and squared-magnitude of the imaginary field satisfy:
\begin{equation}
    \hbar\langle\nabla\cdot p_I\rangle = 2\langle p_I^2\rangle,
    \label{eq:virial-osmotic}
\end{equation}
confirming the internal consistency of the structural and global 
forms.
Substituting Eq.~\eqref{eq:ref4.6} into Eq.~\eqref{eq:Vqu-split}:
\begin{equation}
    \langle V_{\mathrm{qu}}\rangle 
    = \frac{\langle p_I^2\rangle}{m} 
    - \frac{\langle p_I^2\rangle}{2m}
    = \frac{\langle p_I^2\rangle}{2m}.
    \label{eq:Vqu-pI}
\end{equation}
We have derived Eq.~\eqref{eq:Vqu-pI} 
directly from the structural form Eq.~\eqref{eq:ref4.2}, using 
only the definition of $p_I$ and integration by parts --- no 
polar form of $\psi$ is required. 
This is structurally identical to the classical kinetic energy 
formula $\langle T\rangle = \langle p_R^2\rangle/2m$, with $p_I$ 
playing the role of the imaginary momentum. The result is finite 
and non-negative for all $\psi\in H^1(\Omega)$, resolving the 
nodal singularity difficulty of Eq.~\eqref{eq:Bohm-pointwise}. 

The Fisher information connection follows from the density 
identity $p_I = -\frac{\hbar}{2}\nabla\ln\rho$, which gives
$\langle p_I^2\rangle = \frac{\hbar^2}{4}\mathcal{I}[\rho]$,
and hence:
\begin{equation}
    \langle V_{\mathrm{qu}}\rangle
    = \frac{\hbar^2}{8m}\,\mathcal{I}[\rho].
    \label{eq:Vqu-Fisher-info}
\end{equation}
This recovers the known Fisher information 
identity~\cite{Frieden1998,Reginatto1998} through the nonlinear 
decomposition framework. The three equivalent expressions:
\begin{equation}
    \langle V_{\mathrm{qu}}\rangle
    = \frac{\langle p_I^2\rangle}{2m}
    = \frac{\hbar^2}{8m}\,\mathcal{I}[\rho]
    = \frac{\hbar^2}{2m}\int_\Omega |\nabla R|^2\,d\mathbf{r}
    \label{eq:Vqu-triple}
\end{equation}
connect the quantum potential to $\langle p_I^2\rangle/2m$, the 
Fisher information of $\rho$, and the Sobolev norm of $R$. The 
first is the primary result of this section; the third follows 
from $\langle p_I^2\rangle = \hbar^2\int|\nabla R|^2\,d\mathbf{r}$ 
via the polar form and is included for comparison with the 
existing literature~\cite{Heifetz2015}.

\subsubsection{Canonical Example: Infinite Square Well}
\label{sec:ISW}

The result of Eq.~\eqref{eq:Vqu-Fisher-info} can be verified
explicitly for the infinite square well, which provides the most
transparent illustration because the particle velocity vanishes
identically ($p_R = 0$) and the potential is zero inside the well
($V = 0$). The total energy must therefore come entirely from the
quantum potential.

For a particle in an infinite well of width $L$, the normalized
eigenstates are:
\begin{equation}
    \psi_n(x) = \sqrt{\frac{2}{L}}\sin\!\left(\frac{n\pi x}{L}\right),
    \qquad n = 1, 2, 3, \ldots
    \label{eq:ISW-psi}
\end{equation}
Since $\psi_n$ is real, $S = 0$ everywhere and $p_R = \nabla S = 0$,
confirming zero particle velocity. The amplitude is $R(x) =
|\psi_n(x)|$, giving:
\begin{equation}
    |\nabla R|^2 = \frac{2}{L}\cdot\frac{n^2\pi^2}{L^2}
    \cos^2\!\left(\frac{n\pi x}{L}\right).
    \label{eq:ISW-gradR}
\end{equation}
Substituting into the third expression of Eq.~\eqref{eq:Vqu-triple} and using
$\int_0^L\cos^2(n\pi x/L)\,dx = L/2$:
\begin{equation}
    \langle V_{\mathrm{qu}}\rangle_n
    = \frac{\hbar^2}{2m}\cdot\frac{2n^2\pi^2}{L^3}
    \cdot\frac{L}{2}
    = \frac{n^2\hbar^2\pi^2}{2mL^2} = E_n.
    \label{eq:ISW-result}
\end{equation}
The quantum potential accounts for the \emph{entire} energy of every
eigenstate. Since $\langle p_R^2\rangle/2m = 0$ and $\langle V
\rangle = 0$, the energy equation reduces to:
\begin{equation}
    \langle E \rangle = \langle V_{\mathrm{qu}} \rangle = E_n,
    \label{eq:ISW-energy}
\end{equation}
in exact agreement with the standard quantum mechanical result.

This result has a transparent physical interpretation. The energy
$E_n$ is not kinetic in the classical sense --- the Bohmian particle
is at rest. Rather, from Eq.~\eqref{eq:Vqu-pI}, it equals 
$\langle p_I^2\rangle/2m$. The Fisher information of
the $n$-th eigenstate is:
\begin{equation}
    \mathcal{I}[\rho_n] = 4\int_0^L|\nabla R_n|^2\,dx
    = \frac{4n^2\pi^2}{L^2},
    \label{eq:ISW-Fisher}
\end{equation}
confirming $\langle V_{\mathrm{qu}}\rangle_n =
\hbar^2\mathcal{I}[\rho_n]/8m = E_n$ via
Eq.~\eqref{eq:Vqu-Fisher-info}. Higher eigenstates oscillate more
rapidly in space, have larger Fisher information, and consequently
higher energy. The energy quantization of the infinite square well
has a direct information-geometric expression in terms of the Fisher
information of the probability density --- a perspective that is
entirely absent from the pointwise de~Broglie--Bohm formulation.

\subsubsection{Physical and Information-Geometric Interpretation}
\label{sec:Fisher}

The imaginary local field $p_I = -\frac{\hbar}{2}\nabla\ln\rho$ carries a 
precise dual role. In information geometry, it is directly 
proportional to the score function $\nabla\ln\rho$ of the 
probability density $\rho$, whose squared expectation defines 
the classical Fisher information --- linking the decomposition 
to the broader framework of geometric quantum 
mechanics~\cite{Kibble,Ashtekar,Brody}. The connection between 
the quantum potential and Fisher information has been noted 
previously: Reginatto~\cite{Reginatto1998} and Hall and 
Reginatto~\cite{HallReginatto2002} derived the Schr\"odinger 
equation from a variational principle of minimum Fisher 
information, with the quantum potential arising as an 
informational correction at the level of the equation of 
motion; and Heifetz and 
Cohen~\cite{Heifetz2015} observed within the Madelung 
hydrodynamic framework that the integrand of 
$\langle V_{\mathrm{qu}}\rangle$ equals half the squared 
imaginary momentum density. The present work establishes the 
intermediate identity $\langle V_{\mathrm{qu}}\rangle = 
\langle p_I^2\rangle/2m$ as an exact operator-level result, 
proved via Sobolev integration by parts from the decomposition 
of $-i\hbar\nabla$, without any variational principle and for 
all $H^1$ states including those with nodal sets. The known 
Fisher information relation 
$\langle V_{\mathrm{qu}}\rangle = \hbar^2\mathcal{I}[\rho]/8m$ 
then follows as a corollary. In physical terms, 
$p_I$ is the momentum of amplitude deformation: it governs 
the direction along which the wave packet changes shape and 
generates the quantum potential, while $p_R$ governs 
transport --- the direction along which the particle moves and 
drives the probability current.

The quantum potential $V_{\mathrm{qu}}$ from Eq.~\eqref{eq:ref4.2}
is therefore the energy associated with the \emph{curvature of 
the amplitude}. Even if a particle is instantaneously at rest 
(local Bohmian velocity $p_R/m = 0$), it still possesses 
internal energy whenever the amplitude flow has non-zero 
divergence ($\nabla\cdot p_I \neq 0$). This is the geometric 
meaning of the quantum force: it is not an external push, but 
a structural requirement of the curvature of the wavefunction 
amplitude. The structural independence between the 
information-theoretic content of the amplitude (governed by 
$p_I$) and the dynamical transport (governed by $p_R$) parallels 
the decomposition identified by Brody and Hughston~\cite{Brody} 
between the Riemannian and symplectic structures on projective 
Hilbert space. Establishing the precise correspondence between 
this position-space factorization and the abstract quantum 
geometric tensor~\cite{provost1980,berry1984} is a natural 
direction for future work.

\subsubsection{Mathematical Origin of the Quantum Force}

The quantum force $F_{\mathrm{qu}} = -\nabla V_{\mathrm{qu}}$ appears
non-classical because the potential $V_{\mathrm{qu}}$ itself depends on
the local field geometry through $p_I$, rather than being an externally
specified function of position.  In classical mechanics, a force is the
gradient of an external potential. Here, the potential is a functional
of the wavefunction amplitude, making the force inherently non-local in 
the sense that it depends on $p_I$ and its spatial variation at each point.

Writing $V_{\mathrm{qu}}$ in the structural form of 
Eq.~\eqref{eq:ref4.2} makes this origin transparent. The 
local field $p_I/m$ plays the role of a velocity field 
whose divergence governs the compressibility of the 
amplitude flow --- the rate at which the field 
converges or diverges at each point --- while the $p_I^2$ 
term measures its local kinetic intensity. The quantum force is 
the gradient of the net imbalance between these two contributions, 
consistent with the analysis in Ref.~\cite{Khondker2025}.

\section{Canonical Realizations and Boundary Robustness}

The boundary functional derived in Eq.~\eqref{eq:ref3.3} reveals that the vanishing of $\langle p_I \rangle$ is a consequence of the system's boundary topology. We now verify this robustness across several standard quantum systems, demonstrating that the vanishing of the imaginary expectation value is a universal consequence of boundary topology.

\subsection{Localized and Bound States}
For standard systems such as the infinite square well, harmonic oscillator, or particles in Airy-type potentials, the imaginary contribution vanishes due to the behavior of the probability density at the boundaries:
\begin{itemize}
    \item \textbf{Hard Boundaries:} In an infinite well of width $L$, the Dirichlet conditions $\psi(0) = \psi(L) = 0$ ensure that $\rho(0) = \rho(L) = 0$. Consequently, the surface integral $\oint \rho\, \hat{\mathbf{n}}\, dS$ vanishes identically.
    \item \textbf{Decay at Infinity:} For normalizable states in $\mathbb{R}^d$, the requirement $\psi \in L^2$ forces $\rho \to 0$ as $|r| \to \infty$, eliminating the surface term in the limit of an infinite bounding sphere.
\end{itemize}

\subsection{Structural Vanishing in Periodic Systems}

The vanishing of the imaginary component persists under periodic
boundary conditions as a direct consequence of the boundary-functional
identity of Section~3 combined with the lattice periodicity of the
probability density.

Consider a Bloch eigenstate $\psi_{n\mathbf{k}}(\mathbf{r}) =
e^{i\mathbf{k}\cdot\mathbf{r}}u_{n\mathbf{k}}(\mathbf{r})$ on a
primitive cell $\Omega_c$ with Bravais lattice vectors
$\{\mathbf{a}_1, \mathbf{a}_2, \mathbf{a}_3\}$~\cite{AshcroftMermin1976}.
The cell density $\rho(\mathbf{r}) = |u_{n\mathbf{k}}(\mathbf{r})|^2$
satisfies lattice periodicity:
\begin{equation}
    \rho(\mathbf{r} + \mathbf{a}_j) = \rho(\mathbf{r}),
    \quad j = 1,2,3.
    \label{eq:ref5.1}
\end{equation}

\begin{proposition}[Vanishing of the osmotic term]
Let $\rho \in W^{1,1}(\Omega_c)$ satisfy~\eqref{eq:ref5.1}.
Then $\displaystyle\int_{\Omega_c} \rho\, p_I\, d\mathbf{r} = 0$.
\end{proposition}

\begin{proof}
By the density identity~\eqref{eq:ref3.1} and the Gauss--Green
theorem~\cite{ChenTorresZiemer2009}:
\[
    \int_{\Omega_c} \rho\, p_I\, d\mathbf{r}
    \;=\; -\frac{\hbar}{2}\oint_{\partial\Omega_c} \rho\,\hat{n}\,dS.
\]
The boundary $\partial\Omega_c$ consists of opposing face pairs
$\{F_j^+, F_j^-\}$ with $\hat{n}|_{F_j^+} = -\hat{n}|_{F_j^-}$.
Periodicity~\eqref{eq:ref5.1} gives $\rho|_{F_j^+} = \rho|_{F_j^-}$ for each
$j$, so every pair contributes zero flux. The boundary integral
vanishes identically.
\end{proof}

This result is consistent with the general Bloch theorem for current
operators~\cite{watanabe2019}, which establishes that the expectation
value of any $U(1)$ current averaged over a periodic system vanishes
in the thermodynamic limit. The real component $\frac{1}{m}
\int_{\Omega_c}\rho\, p_R\, d\mathbf{r}$ recovers the standard
semiclassical group velocity, including its Berry-phase
corrections~\cite{xiao2010}, while the imaginary component
contributes nothing to band transport --- confirming consistency with the standard theory of Bloch electrons.

\subsection{Periodic Systems and Group Velocity}

The role of the real weak-momentum component becomes particularly
transparent in periodic systems. The purpose of this section is to
demonstrate that the group velocity of a Bloch band is carried
entirely by the real component $p_R$, while the imaginary component $p_I$ contributes nothing --- a consistency check
with the standard band theory that also illustrates the physical
division of labor between the two local fields $p_R$ and $p_I$.

Consider a Bloch eigenstate $\psi_{n\mathbf{k}} =
e^{i\mathbf{k}\cdot\mathbf{r}}u_{n\mathbf{k}}(\mathbf{r})$ with
lattice-periodic $u_{n\mathbf{k}}$~\cite{AshcroftMermin1976}. Here $n$ is the band index and $\mathbf{k}$ is the crystal momentum.
Conjugating the Hamiltonian with $e^{i\mathbf{k}\cdot\mathbf{r}}$ 
yields the shifted operator $H_\mathbf{k} = (\hat{p} + 
\hbar\mathbf{k})^2/2m + V(\mathbf{r})$. The Hellmann--Feynman
theorem~\cite{Feynman1939,Wannier1962,Blount1962} then gives the
standard group velocity:
\begin{equation}
    \nabla_\mathbf{k} E_n(\mathbf{k})
    \;=\; \frac{\hbar}{m}
          \langle u_{n\mathbf{k}} \vert\, \hat{p} + \hbar\mathbf{k}
          \,\vert u_{n\mathbf{k}} \rangle.
    \label{eq:ref5.2}
\end{equation}

To connect this result with the weak-momentum decomposition, note that
from the nonlinear decomposition $\nabla\psi_{n\mathbf{k}} = 
\frac{i}{\hbar}(p_{R,n\mathbf{k}} + ip_{I,n\mathbf{k}})\psi_{n\mathbf{k}}$,
we compute $\psi_{n\mathbf{k}}^*\nabla\psi_{n\mathbf{k}} = 
\frac{i}{\hbar}\rho_{n\mathbf{k}}(p_{R,n\mathbf{k}} + ip_{I,n\mathbf{k}})$.
Since $p_{R,n\mathbf{k}}$ and $p_{I,n\mathbf{k}}$ are real, taking the
imaginary part gives 
$\mathrm{Im}(\psi_{n\mathbf{k}}^*\nabla\psi_{n\mathbf{k}}) = 
\rho_{n\mathbf{k}}\,p_{R,n\mathbf{k}}/\hbar$. The probability current 
density is therefore~\cite{Khondker2025}:
\begin{equation}
\bm{J}_{n\mathbf{k}}(\mathbf{r})
=\frac{\hbar}{m}\,\mathrm{Im}\!\left(\psi_{n\mathbf{k}}^* 
\nabla\psi_{n\mathbf{k}}\right) 
=\frac{\rho_{n\mathbf{k}}(\mathbf{r})\,p_{R,n\mathbf{k}}(\mathbf{r})}{m},
\label{eq:ref5.3} 
\end{equation}
where $p_{R,n\mathbf{k}}(\mathbf{r})/m$ is the local Bohmian velocity 
of the Bloch state. The imaginary field $p_{I,n\mathbf{k}}$ contributes 
only to the real part of $\psi_{n\mathbf{k}}^*\nabla\psi_{n\mathbf{k}}$ 
and therefore drops out of the current entirely. Since $\nabla\psi_{n\mathbf{k}} = 
e^{i\mathbf{k}\cdot\mathbf{r}}(\nabla + i\mathbf{k})
u_{n\mathbf{k}}$, we have $\psi_{n\mathbf{k}}^* \nabla\psi_{n\mathbf{k}} = 
u_{n\mathbf{k}}^*(\nabla + i\mathbf{k})u_{n\mathbf{k}}$.
Integrating the current over the primitive cell $\Omega$:
\begin{align}
    \int_\Omega \bm{J}_{n\mathbf{k}}\,d\mathbf{r}
    &= \frac{\hbar}{m}\int_\Omega
       \mathrm{Im}\!\left[u_{n\mathbf{k}}^*
       (\nabla + i\mathbf{k})u_{n\mathbf{k}}\right]
       d\mathbf{r} \notag\\
    &= \frac{1}{m}
       \langle u_{n\mathbf{k}}\vert\,
       \hat{p} + \hbar\mathbf{k}\,
       \vert u_{n\mathbf{k}}\rangle
    \;=\; \frac{1}{\hbar}
       \nabla_\mathbf{k}E_n(\mathbf{k}),
    \label{eq:ref5.4}
\end{align}
where the last equality uses Eq.~\eqref{eq:ref5.2}. Substituting Eq.~\eqref{eq:ref5.3} into the left side of 
Eq.~\eqref{eq:ref5.4} gives the central result:
\begin{equation}
    \frac{1}{\hbar}\nabla_\mathbf{k} E_n(\mathbf{k})
    \;=\; \int_\Omega
          \frac{\rho_{n\mathbf{k}}(\mathbf{r})\, 
          p_{R,n\mathbf{k}}(\mathbf{r})}{m}
          \,d\mathbf{r}.
    \label{eq:ref5.5}
\end{equation}
The imaginary component does not appear in Eq.~\eqref{eq:ref5.5} 
because, by Proposition~1, its primitive-cell integral vanishes 
identically. The group velocity is therefore a pure phase-gradient 
observable: it is the density-weighted spatial average of the local 
Bohmian velocity $v_R = p_R/m$, with no contribution from $p_I$. This 
confirms that the weak-momentum decomposition is fully consistent 
with the standard semiclassical theory of band 
transport~\cite{xiao2010}, while providing a transparent geometric 
interpretation: the group velocity measures transport along the 
phase direction of the wavefunction, and the amplitude direction is 
dynamically inert for this observable.

\section{Magnetic Fields and Gauge Invariance}
The nonlinear momentum decomposition is naturally compatible with the principle of gauge invariance under minimal coupling to a magnetic field. When a particle of charge $q$ is subject to a magnetic vector potential $\mathbf{A}$, the canonical momentum operator is replaced by the kinetic (mechanical) momentum operator $\hat{\Pi} = -i\hbar\nabla - q\mathbf{A}$.

\subsection{Gauge-Invariant Splitting}

Under minimal coupling, the kinetic momentum operator
$\hat{\Pi} = -i\hbar\nabla - q\mathbf{A}$ acts on $\psi$ as:
\begin{equation}
    \hat{\Pi}\psi 
    \;=\; \mathcal{P}_Q[\psi] - q\mathbf{A}\psi
    \;=\; \bigl(\mathcal{P}_R[\psi] - q\mathbf{A}\psi\bigr) 
          + i\,\mathcal{P}_I[\psi].
    \label{eq:gauge-split}
\end{equation}
Since $q\mathbf{A}$ is real, it shifts only the real transform.
We define the gauge-modified real transform and the 
corresponding local field:
\begin{equation*}
    \mathcal{P}_R^{\mathbf{A}}[\psi] 
    \;\equiv\; \mathcal{P}_R[\psi] - q\mathbf{A}\psi 
    \;=\; p_R^{\mathbf{A}}\psi,   
\end{equation*}
\begin{equation} 
    p_R^{\mathbf{A}} \;\equiv\; p_R - q\mathbf{A},
    \label{eq:PR-magnetic}
\end{equation}
while the imaginary transform $\mathcal{P}_I[\psi] = p_I\psi$ 
and its local field $p_I = \mathrm{Re}(-\hbar\nabla\ln\psi)$ 
are unchanged --- they depend only on $|\psi|^2 = \rho$, 
which is unaffected by a gauge transformation since 
$|e^{iq\Lambda/\hbar}\psi|^2 = |\psi|^2$.
Under the gauge transformation $\psi \to e^{iq\Lambda/\hbar}\psi$ 
and $\mathbf{A} \to \mathbf{A} + \nabla\Lambda$, the logarithmic 
gradient transforms as 
$\nabla\ln\psi \to \nabla\ln\psi + \frac{iq}{\hbar}\nabla\Lambda$.
Since $\frac{iq}{\hbar}\nabla\Lambda$ is purely imaginary, taking 
real and imaginary parts gives:
\begin{align}
    p_R &\;\to\; p_R + q\nabla\Lambda, \notag\\
    p_I &\;\to\; p_I. \notag
\end{align}
The combination $p_R^{\mathbf{A}} = p_R - q\mathbf{A}$ is therefore 
gauge-invariant: the shift $+q\nabla\Lambda$ in $p_R$ exactly 
cancels the shift $-q\nabla\Lambda$ from $-q\mathbf{A}$. The 
field $p_I$ is gauge-invariant because the gauge phase 
$e^{iq\Lambda/\hbar}$ contributes only to the imaginary part 
of $\nabla\ln\psi$, leaving 
$p_I = \mathrm{Re}(-\hbar\nabla\ln\psi)$ unchanged. This 
derivation uses only the definitions of $p_R$ and $p_I$ from 
Eqs.~\eqref{eq:ref2.1}--\eqref{eq:ref2.2} --- no polar form 
is required. The gauge-invariant physical velocity is 
$\mathbf{v} = p_R^{\mathbf{A}}/m$, while the imaginary component 
$p_I/m$ remains independent of gauge choice.

\subsection{Persistence of the Structural Identity}

Since $p_I$ is gauge-invariant, the structural identity
of Section~3 is unaffected by the magnetic field:
\begin{equation}
    \rho\, p_I = -\frac{\hbar}{2}\nabla\rho.
\end{equation}
The expectation value of the imaginary component therefore
reduces to the same boundary functional as in the field-free
case:
\begin{equation}
    \langle p_I \rangle_{\mathbf{A}} 
    = \int_{\Omega} \rho\, p_I \, d\mathbf{r} 
    = -\frac{\hbar}{2} \oint_{\partial\Omega} 
      \rho\, \hat{\mathbf{n}} \, dS = 0,
\label{eq:magnetic_boundary}
\end{equation}
where the last equality follows from the boundary conditions
of Section~5. The magnetic field only modifies $p_R^{\mathbf{A}}$;
the imaginary component $p_I$ contributes nothing to the
kinetic momentum expectation value, independent of the
magnetic flux or gauge choice.

\subsection{Geometric Interpretation}

The magnetic vector potential acts as a $U(1)$ connection. 
The real local field $p_R^{\mathbf{A}} = p_R - q\mathbf{A}$ 
corresponds to the covariant phase gradient determined by 
this connection, while the imaginary local field $p_I$ 
remains a divergence-generated amplitude field independent 
of the connection. Thus, the nonlinear factorization is 
fully compatible with gauge symmetry and preserves the 
geometric structure of the decomposition.

\subsection{Magnetic Energy Decomposition}
Applying the nonlinear decomposition to the magnetic Hamiltonian $\hat{H} = \frac{1}{2m}(\hat{\mathbf{p}} - q\mathbf{A})^2 + V$, the real mathematical operator $e_R$ (defined in the Appendix) yields:
\begin{equation}
    e_R = \frac{(p_R^{\mathbf{A}})^2}{2m} + V + \frac{\hbar \nabla \cdot p_I - p_I^2}{2m}.
\end{equation}
This confirms that the structural form of the quantum potential $V_{\mathrm{qu}}$ is preserved in the presence of magnetic fields, with $p_R$ replaced by the gauge-invariant field $p_R^{\mathbf{A}} = p_R - q\mathbf{A}$. The corresponding imaginary component maintains the continuity relation:
\begin{equation}
    e_I = \frac{\hbar}{2\rho} \left[ \frac{\partial \rho}{\partial t} + \nabla \cdot \left( \rho \frac{p_R^{\mathbf{A}}}{m} \right) \right].
\end{equation}
Since the term in the brackets is the magnetic continuity equation (which vanishes identically for any physical state), the expectation value $\langle e_I \rangle$ remains zero, extending the validity of the vanishing expectation theorem to gauge-coupled systems.

\section{Conclusion}

The results established in this work elevate the nonlinear momentum 
decomposition from a formal reformulation to a rigorous 
functional-analytic framework. By proving that the splitting of 
$-i\hbar\nabla$ into real and imaginary components is well-posed in 
$H^1(\Omega)$, spectrally preserving, and physically interpretable, we establish the decomposition as a legitimate geometric language for 
quantum observables, one that does not modify the spectral content of 
linear theory.

Four structural results follow. First, the vanishing of $\langle p_I 
\rangle$ is a geometric necessity: the imaginary component collapses 
into a divergence-generated boundary functional whose integral is zero 
for all states satisfying standard boundary conditions --- hard-wall 
constraints, decay at infinity, or lattice periodicity --- as 
established through the $W^{1,1}$ Sobolev regularity of $\rho = |\psi|^2$.
Second, the quantum potential $V_{\mathrm{qu}}$ of Eq.~\eqref{eq:ref4.2}
arises entirely from $p_I$, and its expectation value satisfies 
the exact identity $\langle V_{\mathrm{qu}}\rangle = 
\langle p_I^2\rangle/2m$, proved via Sobolev integration by 
parts without invoking the polar form and for all $H^1$ states 
including those with nodal sets. The known Fisher information 
identity $\langle V_{\mathrm{qu}}\rangle = 
\hbar^2\mathcal{I}[\rho]/8m$~\cite{Frieden1998,Reginatto1998} 
follows as a corollary. Third, the group velocity of 
a Bloch band is the density-weighted spatial average of the local 
Bohmian velocity $p_R/m$, with the imaginary component $p_I$ 
contributing nothing to band transport, in full consistency with 
semiclassical theory~\cite{xiao2010}. Fourth, under magnetic minimal 
coupling, the vector potential $\mathbf{A}$ acts exclusively on the 
phase component $p_R$, leaving $p_I$ as a purely geometric attribute 
of the amplitude, independent of gauge choice.

A further consequence is the natural resolution of nodal singularities. 
The $\rho = |\psi|^2$ weighting cancels the logarithmic divergence of $p_I$ 
at nodal sets $\{\psi=0\}$, rendering the imaginary momentum 
$L^1$-integrable and geometrically benign. This resolves a difficulty 
that has historically complicated trajectory-based interpretations of 
quantum mechanics without requiring any restriction to non-vanishing 
wavefunctions.

As a verification, for the infinite square well where 
$p_R = 0$ and $V = 0$, the identity 
$\langle V_{\mathrm{qu}}\rangle = \hbar^2\mathcal{I}[\rho]/8m$ 
gives $\langle E\rangle = \langle p_I^2\rangle/2m = E_n$ 
exactly. The present derivation makes transparent why: the 
energy quantization of the infinite square well has a direct 
information-geometric expression in terms of the Fisher 
information of the probability density, a perspective 
absent from the pointwise de~Broglie--Bohm formulation.

In the language of geometric quantum mechanics~\cite{Kibble,Ashtekar,
Brody}, $p_R$ governs transport --- Bohmian trajectories, probability 
current, and band velocity --- while $p_I$ governs amplitude deformation 
--- its energy cost and the quantum potential. The 
nonlinear factorization does not alter the spectral content of quantum 
mechanics. Instead, it provides an intrinsic, gauge-invariant language, 
independent of any particular wavefunction representation, that makes 
this geometric division of labor explicit, exposing a structure that is 
present in, but entirely obscured by, the standard linear formulation 
of the momentum operator.

\paragraph{Outlook.}
Equation~\eqref{eq:refA.22} has a natural extension beyond closed Hermitian systems. In a closed Hermitian system, $e_I(\mathbf{r},t)$ redistributes probability locally but its density-weighted integral vanishes --- as follows from Eq.~\eqref{eq:refA.norm} and the conservation of the global norm by unitarity --- so the global density is conserved. Locally, however, $e_I(\mathbf{r},t)$ governs the flow of probability density into or out of a volume element and need not vanish pointwise, acting as a local source-sink for the probability current. In open quantum systems described by effective non-Hermitian Hamiltonians~\cite{moiseyev2011,rotter2015}, even the density-weighted integral of $e_I$ need not vanish, and Eq.~\eqref{eq:refA.22} then predicts qualitatively distinct density dynamics: decay where $e_I < 0$ (probability absorption, resonance states with finite lifetimes) and growth where $e_I > 0$ (stimulated amplification, unstable resonance modes). A rigorous treatment would require specifying the non-Hermitian Hamiltonian explicitly --- for example, a complex optical potential $V(\mathbf{r}) = V_{\mathrm{re}}(\mathbf{r}) + iV_{\mathrm{im}}(\mathbf{r})$ with $V_{\mathrm{im}} < 0$ for absorbing regions~\cite{moiseyev2011} --- and deriving the corresponding $e_I$ from first principles within the nonlinear decomposition framework. We leave this extension to future work.

\section*{Acknowledgments}
AI-assisted tools were used for language editing 
and manuscript preparation.

\appendix
\section{Appendix}
\section*{Hydrodynamic Correspondence and Local Energy Decomposition}\label{sec:AppendixA}

\subsection{Expansion of the Mathematical Term $e_R(t)$}
Following Ref.~\cite{Khondker2025}, we introduce two mathematical operators, $e_R(\mathbf{r},t)$ and $e_I(\mathbf{r},t)$, respectively, as follows:
\begin{align}
e_R(\mathbf{r},t) \{\psi(\mathbf{r},t)\} & \equiv \Re e( i \hbar \frac{\partial}{\partial t} \ln \psi(\mathbf{r},t) ) \psi(\mathbf{r},t) \label{eq:refA.1}\\
e_I(\mathbf{r},t) \{\psi(\mathbf{r},t)\} & \equiv \Im m( i \hbar \frac{\partial}{\partial t} \ln \psi(\mathbf{r},t) ) \psi(\mathbf{r},t) \label{eq:refA.2}
\end{align}
It is easy to verify that
\begin{equation}
e_R \psi + i e_I \psi = i \hbar \frac{\partial }{\partial t} \psi = \hat{H}\psi \label{eq:refA.3}
\end{equation}

The evolution of the phase is governed by the real-valued term $e_R(t)$, defined via the local action of the Hamiltonian $\hat{H} = \frac{\mathcal{P}_Q^2}{2m} + V$:
\begin{equation}
    e_R(t) = \text{Re}\left( \frac{\hat{H}\psi}{\psi} \right) = \text{Re}\left( \frac{\mathcal{P}_Q^2 \psi}{2m\psi} \right) + V
\end{equation}
Evaluating the second-order operator $\mathcal{P}_Q^2$ by recursive application of the nonlinear maps:
\begin{align}
    \mathcal{P}_Q^2 \psi &= -i\hbar\nabla \cdot [(p_R + ip_I)\psi] \notag \\ 
    &= -i\hbar[(\nabla \cdot (p_R + i p_I))\psi + (p_R + ip_I) \cdot \nabla\psi]  
\end{align}
Substituting the identity $\nabla\psi = \frac{i}{\hbar}(p_R + ip_I)\psi$, the expansion becomes:
\begin{equation}
    \mathcal{P}_Q^2 \psi = -i\hbar(\nabla \cdot p_R + i\nabla \cdot p_I)\psi + (p_R + ip_I)^2 \psi
\end{equation}
Isolating the real part to define $e_R(t)$, we obtain:
\begin{equation}
    e_R(t)\psi  = \left[ \frac{p_R^2}{2m} + V + \frac{\hbar\,\nabla\cdot p_I \;-\; p_I^2}{2m}\right]\psi 
\end{equation}

Similarly, isolating the imaginary part of the same expansion 
yields the imaginary mathematical operator. Noting that 
$(p_R + ip_I)^2 = p_R^2 - p_I^2 + 2ip_R\cdot p_I$ and 
$-i\hbar(\nabla\cdot p_R + i\nabla\cdot p_I) = 
\hbar\nabla\cdot p_I - i\hbar\nabla\cdot p_R$, the imaginary 
part of $\mathcal{P}_Q^2\psi/(2m\psi)$ gives:
\begin{equation}
    e_I(t)\psi \;=\; \left[\frac{2p_R\cdot p_I 
    \;-\; \hbar\,\nabla\cdot p_R}{2m}\right]\psi.
    \label{eq:refA.eI}
\end{equation}

\subsection{Identification of the Quantum Potential}\label{sec:AppendixA.2}
Consistent with the quantum Hamilton-Jacobi equation, $e_R(t)$ is partitioned into the classical-like kinetic energy plus the potential energy $V$ and the structural quantum potential $V_{qu}$~\cite{Bohm1952}:
\begin{equation}
    e_R(t)\psi  = \left[\frac{p_R^2}{2m} + V + V_{qu} \right]\psi 
\end{equation}
Here we obtain the structural form of the quantum potential using only the imaginary local field $p_I$:
\begin{equation}
        V_{\mathrm{qu}} \;=\; \frac{\hbar\,\nabla\cdot p_I \;-\; p_I^2}{2m}.  \label{eq:refA.10}
\end{equation}
We use the ensemble average energy definition from~\cite{Khondker2025}:
\begin{equation}
\langle E \rangle_{\text{ensemble average}} = \int_{-\infty}^{\infty} d\mathbf{r}\, \psi^2 e_R(\mathbf{r}, t)
\end{equation}
Substituting our expression for $e_R$ yields the Hamilton-Jacobi equation:
\begin{equation}
\int_{-\infty}^{\infty} d\mathbf{r}\, \psi^2 e_R(t) = \int_{-\infty}^{\infty} d\mathbf{r}\, \psi^2 \left[ \frac{p_R^2}{2m} + V(\mathbf{r}) + \left( \frac{\hbar \nabla \cdot p_I - p_I^2}{2m} \right) \right] \label{eq:refA.12}
\end{equation}

The above equation corresponds to the reformulation of the de~Broglie--Bohm Hamilton-Jacobi equation~\cite{Bohm1952}. Using Eq.~\eqref{eq:ref2.6} one arrives at~\cite{Khondker2025}:
\begin{equation}
 \int_{-\infty}^{\infty} d\mathbf{r} \, \psi^2  \left[-\frac{\partial S}{\partial t} \right] =  \int d\mathbf{r} \, \psi^2 \left[\frac{(\nabla S)^2}{2m} + V(\mathbf{r}) - \frac{\hbar^2}{2m} \frac{\nabla^2 R}{R}\right] \label{eq:refA.13}
\end{equation}
Steps are shown below.


From Eq.~\eqref{eq:ref2.6}, the imaginary local field is:
\begin{equation}
    p_I \;=\; -\hbar\,\nabla\ln R \;=\; -\frac{\hbar\,\nabla R}{R}.
\label{eq:refA.14}   
\end{equation}
Its squared magnitude is therefore:
\begin{equation}
    p_I^2 \;=\; \frac{\hbar^2\,|\nabla R|^2}{R^2}.
\label{eq:refA.15}   
\end{equation}


Applying the divergence to Eq.~\eqref{eq:refA.14}:
\begin{equation}
    \nabla\cdot p_I
    \;=\; -\hbar\,\nabla\cdot\!\left(\frac{\nabla R}{R}\right)
    \;=\; -\hbar\left(\frac{\nabla^2 R}{R}
          - \frac{|\nabla R|^2}{R^2}\right).
    \label{eq:refA.16}   
\end{equation}
This uses the standard vector identity
$\nabla\cdot(f\mathbf{g}) = f\,\nabla\cdot\mathbf{g} + \mathbf{g}\cdot\nabla f$
with $f = R^{-1}$ and $\mathbf{g} = \nabla R$, giving
$\nabla\cdot({\nabla R}/{R}) = {\nabla^2 R}/{R} - {|\nabla R|^2}/{R^2}$.

Inserting Eqs.~\eqref{eq:refA.15} and \eqref{eq:refA.16} into the structural forms Eq.~\eqref{eq:ref4.2} or \eqref{eq:refA.10}:
\begin{align}
    V_{\mathrm{qu}}
    &= -\frac{\hbar^2}{2m}\frac{\nabla^2 R}{R}.
    \label{eq:refA.17}   
\end{align}
The $|\nabla R|^2/R^2$ terms cancel exactly, and we recover precisely
$V_{\mathrm{qu}}^{\mathrm{Bohm}}$. This confirms that the structural form Eq.~\eqref{eq:refA.10} is an \emph{exact algebraic identity} with Bohm's original expression.

\subsection{Spatial and time derivative of Density}
From the definition of the imaginary local field 
$p_I = \operatorname{Im}\!\left(\tfrac{\hbar}{i}\nabla\ln\psi\right) 
= \operatorname{Re}(-\hbar\nabla\ln\psi)$~\cite{Khondker2025,Takabayasi},
we decompose $\nabla\ln\psi$ into real and imaginary parts. Since 
$\psi = |\psi|e^{i\arg(\psi)}$, we have 
$\nabla\ln\psi = \nabla\ln|\psi| + i\,\nabla\arg(\psi)$, 
and therefore:
\begin{equation}
    \operatorname{Re}(\nabla\ln\psi) 
    = \nabla\ln|\psi| 
    = \nabla\ln\sqrt{\rho} 
    = \frac{\nabla\rho}{2\rho}.
    \label{eq:refA.density-log}
\end{equation}
Multiplying by $-\hbar$ gives $p_I = -\frac{\hbar}{2}\frac{\nabla\rho}{\rho}$.
This step uses only $\rho = |\psi|^2$ and the definition of $p_I$ 
--- no polar form of $\psi$ is required.
Rearranging gives the fundamental geometric identity:
\begin{equation}
    \nabla \rho = -\frac{2}{\hbar} \rho p_I. \label{eq:refA.density-id}
\end{equation}
This identity confirms that $p_I$ is the logarithmic gradient 
of the probability density, characterizing the direction of 
steepest change in amplitude.

Equation~\eqref{eq:refA.density-id} shows that $p_I$ is a pure divergence term when weighted by the density: $\rho p_I = -\frac{\hbar}{2}\nabla\rho$. Since the integral of a pure gradient vanishes for all physically admissible states (as established in Section~3), this identity explains structurally why the imaginary component of the momentum does not contribute to the global expectation value, even though it is essential for the local structure of the quantum potential that appears in the equations of Bohm~\cite{Bohm1952} and Takabayasi~\cite{Takabayasi}.


Starting from $\rho = \psi^*\psi$ and applying the product rule:
\begin{equation}
    \frac{\partial\rho}{\partial t}
    \;=\; \dot{\psi}^*\psi + \psi^*\dot{\psi}.
    \label{eq:refA.20}  
\end{equation}
From Eq.~\eqref{eq:refA.3}, the Schr\"odinger equation gives 
$i\hbar\dot{\psi} = (e_R + ie_I)\psi$, so:
\begin{equation}
    \dot{\psi} \;=\; \frac{e_R + ie_I}{i\hbar}\,\psi,
    \qquad
    \dot{\psi}^* \;=\; \frac{e_R - ie_I}{-i\hbar}\,\psi^*,
    \label{eq:refA.21} 
\end{equation}
where the second expression follows by complex conjugation, 
using the fact that $e_R$ and $e_I$ are real. Substituting 
into Eq.~\eqref{eq:refA.20}:
\begin{equation}
    \frac{\partial\rho}{\partial t}
    \;=\; \frac{\rho}{i\hbar}
          \left[-(e_R - ie_I) + (e_R + ie_I)\right]
    \;=\; \frac{2e_I(t)}{\hbar}\,\rho.
    \label{eq:refA.22} 
\end{equation}
This derivation uses only the definition $\rho = \psi^*\psi$ 
and Eq.~\eqref{eq:refA.3} --- no polar form of $\psi$ is 
required. It shows that $e_I(\mathbf{r},t)$ governs the 
local rate of change of probability density at each point --- 
it acts as a local source-sink, redistributing probability 
between volume elements through the osmotic flow. For any 
physical state governed by a Hermitian Hamiltonian, $e_I$ 
need not vanish pointwise. However, integrating 
Eq.~\eqref{eq:refA.22} over all space gives:
\begin{equation}
    \frac{d}{dt}\int_\Omega\rho\,d\mathbf{r} 
    \;=\; \frac{2}{\hbar}\int_\Omega \rho\, e_I\,d\mathbf{r}.
    \label{eq:refA.norm}
\end{equation}
For a closed Hermitian system, the global density is conserved 
by unitarity: $\frac{d}{dt}\int_\Omega\rho\,d\mathbf{r} = 0$. 
Therefore $\int_\Omega\rho\, e_I\,d\mathbf{r} = 0$, consistent 
with the vanishing of the density-weighted average of $e_I$ 
for all physically admissible states.

The imaginary component of the local Hamiltonian action, denoted as $e_I(t)$, determines the time evolution of the density $\rho = \psi^*\psi$. Substituting Eq.~\eqref{eq:refA.eI} into Eq.~\eqref{eq:refA.22}:
\begin{equation}
    \frac{\partial \rho}{\partial t} = \frac{2\rho}{\hbar} e_I(t) = \frac{\rho}{m\hbar} \left( 2p_R \cdot p_I - \hbar \nabla \cdot p_R \right)
    \label{eq:refA.cont1}
\end{equation}
Applying the product rule and substituting 
$\nabla\rho = -\frac{2}{\hbar}\rho\,p_I$ from 
Eq.~\eqref{eq:refA.density-id}:
\begin{align}
    \nabla\cdot\!\left(\rho\frac{p_R}{m}\right)
    &= \frac{1}{m}\left(\rho\,\nabla\cdot p_R 
       + p_R\cdot\nabla\rho\right) \notag\\
    &= \frac{1}{m}\left(\rho\,\nabla\cdot p_R 
       - \frac{2\rho}{\hbar}\,p_R\cdot p_I\right) \notag\\
    &= \frac{\rho}{m\hbar}
       \left(\hbar\,\nabla\cdot p_R 
       - 2p_R\cdot p_I\right).
    \label{eq:refA.divJ}
\end{align}
This confirms that the time-dependent density is conserved purely through transport driven by the real local field $p_R$:
\begin{equation}
    \frac{\partial \rho}{\partial t} + \nabla \cdot \left( \rho \frac{p_R}{m} \right) = 0
\end{equation}

\bibliographystyle{unsrt}
\bibliography{weak_momentum_arxiv}

\end{document}